\documentclass[10pt]{iopart}

\usepackage{iopams}
\usepackage{setstack}
\usepackage{graphicx}

\newtheorem{theorem}{Theorem}
\newtheorem{lemma}[theorem]{Lemma}
\newtheorem{proposition}[theorem]{Proposition}
\newenvironment{proof}[1][Proof]{\noindent\textbf{#1.} }{\ \rule{0.5em}{0.5em}}

\begin{document}

\title{Fleming's bound for the decay of mixed states}
\author{Florian Fr\"{o}wis, Gebhard Gr\"{u}bl, and Markus Penz}
\address{Institute for Theoretical Physics\\University
of Innsbruck\\Technikerstr. 25, A-6020 Innsbruck, Austria}
\ead{gebhard.gruebl@uibk.ac.at}

\begin{abstract}
Fleming's inequality is generalized to the decay function of mixed
states. We show that for any symmetric hamiltonian $h$ and for any
density operator $\rho$ on a finite dimensional Hilbert space with
the orthogonal projection $\Pi$ onto the range of $\rho$ there
holds the estimate $\Tr(\Pi \rme^{-\rmi ht}\rho \rme^{\rmi ht})
\geq\cos^{2}( ( \Delta h) _{\rho}t)  $ for all real $t$ with $(
\Delta h) _{\rho }\left\vert t\right\vert \leq\pi/2.$ We show that
equality either holds for all $t\in\mathbb{R}$ or it does not hold
for a single $t$ with $0<( \Delta h) _{\rho}\left\vert
t\right\vert \leq\pi/2.$ All the density operators saturating the
bound for all $t\in\mathbb{R},$ i.e. the mixed intelligent states,
are determined.
\end{abstract}


\section{Introduction}

Two states $\rho_{1}$ and $\rho_{2}$ of a quantum system can be
discriminated on the basis of a single measurement outcome if
there exists an observable $A$ such that the probability measures
which are generated by $\rho_{1}$ and $\rho_{2}$ on the spectrum
of $A$ have disjoint supports. In particular if a state $\rho$
evolves under a Hamiltonian $H$ into the state $\rho_{t}$ it may
be desirable to determine and perhaps to minimize a time $t>0$
when the evolved state $\rho_{t}$ can be discriminated from the
initial state $\rho$ by a single measurement. A more realistic
goal is to distinguish $\rho_{t}$ from $\rho$ by performing single
measurements on \textquotedblleft few\textquotedblright\ ensemble
members only. If one chooses as observable $A$ an orthogonal
projection $\Pi$ with $\Tr\left(  \Pi\rho\right)  =1,$ then this
can be done if $\Tr\left( \Pi\rho_{t}\right)  $ is close to $0$
since this means that it is very unlikely to find the property
$\Pi$ in the state $\rho_{t},$
while it is certain in the state $\rho.$\\

Under a somewhat broader perspective the quantity $\Tr\left(
\Pi\rho _{t}\right)  $ is commonly used in order to formalize the
intuitive picture of the decay of a property.~\cite{Exn} Since in
many cases the survival probability $\Tr\left(  \Pi\rho_{t}\right)
$ of the property $\Pi$ cannot be computed explicitly, there
arises the quest for estimates of the decay-function\
$P_{\rho}:t\mapsto \Tr\left(
\Pi\rho_{t}\right).$\\

An important such estimate for $P_{\rho}$ in the case of a pure
state $\rho$ and in case of the property $\Pi=\rho$ is due to
Mandelstam and Tamm \cite{MaT}. This estimate has been
rediscovered by a different reasoning almost 30 years later by
Fleming.~\cite{Fle} Since then it is called Fleming's bound. It
says that for any pure state $\rho$ with a finite energy
uncertainty
$\left(  \Delta H\right)  _{\rho}$ there holds%
\begin{equation}
P_{\rho}\left(  t\right)  \geq\cos^{2}\frac{\left(  \Delta H\right)  _{\rho}%
t}{\hbar}\mbox{ for all }t\mbox{ with }\frac{\left(  \Delta H\right)
_{\rho
}\left\vert t\right\vert }{\hbar}\leq\pi/2.\label{Flemb}%
\end{equation}
From the estimate (\ref{Flemb}) a lower bound to any positive $t$
such that $\Tr\left(  \Pi\rho_{t}\right)  =\varepsilon$ is
obvious\footnote{Clearly this does not imply that there exists any
$t$ at all such that $P_{\rho}\left( t\right)  =\varepsilon$
holds.}:
\[
\frac{\hbar}{\left(  \Delta H\right)  _{\rho}}\arccos\sqrt{\varepsilon}\leq t.
\]
The special case $\varepsilon=0$ leads to the inequality
\begin{equation}
\frac{\pi\hbar}{2\left(  \Delta H\right)  _{\rho}}\leq t\label{otimelb}%
\end{equation}
for the smallest time $t>0$ with $\Tr\left(  \Pi\rho_{t}\right)
=0.$ This time, if existent, is called orthogonalization
\cite{GLM} or passage time \cite{Bro}. Clearly it would also be
useful to have an upper bound for $P_{\rho},$ from which the
existence of an orthogonalization time could be inferred.
Polynomial upper bounds have been given by Andrews \cite{And},
which, however, are strictly positive. Therefore they do not yield
an upper bound to an orthogonalization time.\\

A simple geometric meaning of Fleming's bound became clear through
the time-energy uncertainty relation of Aharonov and Anandan
\cite{AnA}: First, $2t\left( \Delta H\right)  _{\rho}/\hbar$
equals the arc length of the curve $\lambda\mapsto\rho_{\lambda}$
with $0\leq\lambda\leq t$ in the projective space
$\mathcal{P}\left( \mathcal{H}\right)  $ of one dimensional
subspaces of $\mathcal{H}.$ Second, $2\arccos\sqrt{P_{\rho}\left(
t\right) }$ equals the geodesic distance between $\rho$ and
$\rho_{t}$ in $\mathcal{P}\left( \mathcal{H}\right).$ Here the
Riemannian geometry is defined by the Fubini-Study metric of
$\mathcal{P}\left( \mathcal{H}\right)  .$ Thus, as has been
pointed out by Brody \cite{Bro}, Fleming's bound (\ref{Flemb}) is
equivalent to the fact that the length of a curve in
$\mathcal{P}\left(  \mathcal{H}\right) $ is not less than the
geodesic distance between its initial and end
point.\\

In Ref. \cite{HoM} for given Hamiltonian $H$ all pure states
$\rho$ with an orthogonalization time equal to the lower bound
$\pi\hbar/2\left(  \Delta H\right)  _{\rho}$ of equation
(\ref{otimelb}) have been identified. I.e., for such states there
holds $\left(  \Delta H\right)  _{\rho}t=\hbar\pi/2$ for the
smallest $t>0$ with $P_{\rho}\left(  t\right)  =0.$ These states
are called \textquotedblleft intelligent states\textquotedblright\
as they saturate the Aharonov-Anandan uncertainty relation. A pure
state $\rho$ is found to be intelligent if and only if there exist
two eigenvectors $\phi_{1},\phi_{2}$ of $H$ corresponding to
different eigenvalues and with $\left\Vert \phi _{1}\right\Vert
=\left\Vert \phi_{2}\right\Vert $ such that $\rho$ equals the
orthogonal projection onto the one dimensional subspace $\mathbb{C}%
\cdot\left(  \phi_{1}+\phi_{2}\right)  .$ \cite{HoM}\\

In \cite{Bro} a lower bound for the smallest $t>0$ with
$P_{\rho}\left( t\right)  =0$ is given for a special type of a mixed
states: The density operator $\rho$ is assumed to be a mixture of
mutually orthogonal intelligent pure states. In Ref. \cite{GLM} a
generalization of the orthogonalization time to mixed states has
been addressed too. In this work, however, the fidelity
$\Tr\sqrt{\sqrt{\rho}\rho_{t}\sqrt{\rho}}$ is used as a measure of
the degree of decay. Clearly, for mixed states the fidelity does not
coincide with $\Tr\left(  \Pi\rho_{t}\right)  .$ So neither of the
two works \cite{Bro}, and \cite{GLM} presents a generalization of
Fleming's bound (\ref{Flemb}) to the case of an arbitrary mixed
state.\\

In this paper we generalize Fleming's bound to an arbitrary mixed
state $\rho.$ We consider the decay function $P_{\rho}\left(
t\right)  =\Tr\left( \Pi\rho_{t}\right)  ,$ where $\Pi$ is chosen to
be the orthogonal projection onto the range of $\rho$ and we confine
ourselves to finite dimensional Hilbert spaces. We first extend
Fleming's bound to $P_{\rho}.$ We then sharpen the bound by proving
that only one of the two cases

\begin{enumerate}
\item $P_{\rho}\left(  t\right)  >\cos^{2}\left(  \left(  \Delta H\right)
_{\rho}t/\hbar\right)  $ for all $t$ with $0<\left(  \Delta H\right)  _{\rho
}\left\vert t\right\vert /\hbar\leq\pi/2$

\item $P_{\rho}\left(  t\right)  =\cos^{2}\left(  \left(  \Delta H\right)
_{\rho}t/\hbar\right)  $ for all $t\in\mathbb{R}$
\end{enumerate}

$\mathbb{\ }$is realized. Then we identify the set of all density operators
which saturate Fleming's bound. In order to have the paper reasonably
selfcontained we have included a treatment of some closely related well known
results on pure state decay. In this way it also becomes more visible which
structures remain unchanged when going from pure states to mixed ones. The
pure state decay function $P_{\rho}$ is denoted as $P_{\phi}$ when $\rho
=\phi\left\langle \phi,\cdot\right\rangle .$

\section{Pure state decay}

Let $\mathcal{H}$ be a finite dimensional Hilbert space. The
scalar product of two vectors $\phi,\psi\in\mathcal{H}$ is denoted
by $\left\langle \phi ,\psi\right\rangle .$ Let the dynamics of
$\mathcal{H}$ be given in terms of a symmetric linear operator
$h:\mathcal{H}\rightarrow\mathcal{H}$ by $\phi _{t}=\exp\left(
-\rmi ht\right)  \phi$ for $t\in\mathbb{R}$ and $\phi
\in\mathcal{H}.$ The survival amplitude
$A_{\phi}:\mathbb{R}\rightarrow \mathbb{C}$ is defined for
$\phi\in\mathcal{H}$ with $\left\Vert \phi\right\Vert =1$ through
$A_{\phi}\left(  t\right)  =\left\langle \phi
,\phi_{t}\right\rangle $ and accordingly the survival probability
of $\phi$ as a function of $t$ is given by $P_{\phi}=\left\vert
A_{\phi}\right\vert ^{2}:\mathbb{R}\rightarrow\mathbb{R}_{\geq0}.$
From the Cauchy-Schwarz inequality we have $P_{\phi}\leq1.$ The
nonnegative number $P_{\phi}\left( t\right)  $ is the probability
that the pure state $\phi_{t}\left\langle
\phi_{t},\cdot\right\rangle $ passes a preparatory filter for the
state $\phi\left\langle \phi,\cdot\right\rangle .$ Due to
\[
A_{\phi}\left(  -t\right)  =\overline{A_{\phi}\left(  t\right)  }%
\]
$P_{\phi}$ is an even function. Since $\phi_{0}=\phi$ we have
$A_{\phi}\left( 0\right)  =1=P_{\phi}\left(  0\right)  .$\\

The expectation value of $h$ in the state $\phi\left\langle \phi
,\cdot\right\rangle $ is denoted by $\left\langle h\right\rangle _{\phi
}=\left\langle \phi,h\phi\right\rangle $ and its variance reads%
\[
\left(  \Delta h\right)  _{\phi}^{2}=\left\langle h^{2}\right\rangle _{\phi
}-\left\langle h\right\rangle _{\phi}^{2}.
\]
$\phi$ is an eigenvector of $h$ if and only if $\left(  \Delta h\right)
_{\phi}=0.$ Thus for $\left(  \Delta h\right)  _{\phi}=0$ the function
$P_{\phi}$ is constant, i.e. $P_{\phi}\left(  t\right)  =1$ holds for all $t.$
For $\left(  \Delta h\right)  _{\phi}>0,$ however, $P_{\phi}$ is not constant
since for $t\rightarrow0$%
\begin{eqnarray*}
P_{\phi}\left(  t\right)   & =\left\vert 1-\rmi t\left\langle
h\right\rangle
_{\phi}-\frac{1}{2}t^{2}\left\langle h^{2}\right\rangle _{\phi}+\rmi\frac{1}%
{3!}t^{3}\left\langle h^{3}\right\rangle _{\phi}+\Or\left(
t^{4}\right)
\right\vert ^{2}\\
& =\left(  1-\frac{1}{2}t^{2}\left\langle h^{2}\right\rangle _{\phi}\right)
^{2}+\left(  t\left\langle h\right\rangle _{\phi}-\frac{1}{3!}t^{3}%
\left\langle h^{3}\right\rangle _{\phi}\right)  ^{2}+\Or\left(  t^{4}\right) \\
& =1-\left(  \Delta h\right)  _{\phi}^{2}t^{2}+\Or\left(
t^{4}\right)  .
\end{eqnarray*}
Thus $P_{\phi}$ has a strict local maximum at $t=0$ if and only if
$\left( \Delta h\right)  _{\phi}>0.$\\

Due to the spectral theorem there exist (unique) nonzero pairwise orthogonal
vectors $\phi_{1},\ldots\phi_{n}$ with $h\phi_{\alpha}=\omega_{\alpha}%
\phi_{\alpha}$ and $\omega_{1}<\ldots<\omega_{n}$ such that
\[
\phi_{t}=\rme^{-\rmi \omega_{1}t}\phi_{1}+\ldots+\rme^{-\rmi \omega_{n}t}\phi_{n}%
\]
for all $t.$ Then $A_{\phi}\left(  t\right)  =\sum_{\alpha=1}^{n}%
\lambda_{\alpha}\rme^{-\rmi \omega_{\alpha}t}$ with
$\lambda_{\alpha}=\left\Vert \phi_{\alpha}\right\Vert ^{2}>0$
follows. For $P_{\phi}\left(  t\right)  $ one
obtains%
\begin{equation}
P_{\phi}\left(  t\right)  =\sum_{\alpha,\beta=1}^{n}\lambda_{\alpha}%
\lambda_{\beta}\rme^{-\rmi \left(  \omega_{\alpha}-\omega_{\beta}\right)  t}%
=\sum_{\alpha,\beta=1}^{n}\lambda_{\alpha}\lambda_{\beta}\cos\left[  \left(
\omega_{\alpha}-\omega_{\beta}\right)  t\right]  .\label{spectrep}%
\end{equation}
Thus both $A_{\phi}$ and $P_{\phi}$ are the restriction of an entire
function to the real line. In particular $A_{\phi}$ and $P_{\phi}$
are $C^{\infty}$ functions.\\

It has been shown by Mandelstam and Tamm \cite{MaT}, and along a different
strategy by Fleming in \cite{Fle} that for all $t$ with $\left(  \Delta
h\right)  _{\phi}\left\vert t\right\vert \leq\pi/2$ there holds%
\[
P_{\phi}\left(  t\right)  \geq\cos^{2}\left(  \left(  \Delta h\right)  _{\phi
}t\right)  .
\]
The original proof of Mandelstam and Tamm \cite{MaT} has been
elaborated by Schulmann in \cite{Sch}. A new proof has been given
recently by Kosi\'{n}ski and Zych. \cite{KoZ}\\

We shall now prove the following somewhat stronger result
implicitly contained in \cite{GLM} and \cite{Bro}.

\begin{proposition}
\label{Prop1}Let $\phi\in\mathcal{H}$ with $\left\Vert
\phi\right\Vert =1$ and $\left(  \Delta h\right)_{\phi}>0.$ Then
exactly one of the alternatives (i) and (ii) holds.

\begin{enumerate}
\item $P_{\phi}\left(  t\right)  >\cos^{2}\left(  \left(  \Delta
h\right) _{\phi}t\right)  $ for all $t\in\mathbb{R}$ with $0<\left(
\Delta h\right) _{\phi}\left\vert t\right\vert \leq\pi/2$

\item $P_{\phi}\left(  t\right)  =\cos^{2}\left(  \left(  \Delta
h\right) _{\phi}t\right)  $ for all $t\in\mathbb{R}$
\end{enumerate}

Alternative (ii) holds if and only if there exist two vectors $\phi_{1}%
,\phi_{2}\in\mathcal{H}$ with $h\phi_{i}=\omega_{i}\phi_{i},\ \omega
_{1}<\omega_{2},\ \left\Vert \phi_{i}\right\Vert ^{2}=1/2$ such that
$\phi=\phi_{1}+\phi_{2}.$
\end{proposition}

\begin{proof}
Let $\Pi=\phi\left\langle \phi,\cdot\right\rangle .$ Then holds
$P_{\phi }\left(  t\right)  =\left\langle \phi,\rme^{\rmi ht}\Pi
\rme^{-\rmi ht}\phi\right\rangle
=\left\langle \Pi\right\rangle _{\phi_{t}}.$ From this it follows that%
\[
\frac{\rmd}{\rmd t}P_{\phi}\left(  t\right)  =\rmi\left\langle
\phi,\rme^{\rmi ht}\left[ h,\Pi\right] \rme^{-\rmi
ht}\phi\right\rangle =\rmi\left\langle \left[  h,\Pi\right]
\right\rangle _{\phi_{t}}.
\]
Using the uncertainty relation for the pair $\left(  h,\Pi\right)  $
we thus obtain for $P_{\phi}^{\prime}\left(  t\right)
=\frac{\rmd}{\rmd t}P_{\phi}\left(
t\right)  $ the estimate%
\[
\left\vert P_{\phi}^{\prime}\left(  t\right)  \right\vert =\left\vert
\left\langle \left[  h,\Pi\right]  \right\rangle _{\phi_{t}}\right\vert
\leq2\left(  \Delta h\right)  _{\phi}\left(  \Delta\Pi\right)  _{\phi_{t}}.
\]
From $\left(  \Delta\Pi\right)  _{\phi_{t}}^{2}=\left\langle \Pi
^{2}\right\rangle _{\phi_{t}}-\left\langle \Pi\right\rangle _{\phi_{t}}%
^{2}=\left\langle \Pi\right\rangle _{\phi_{t}}-\left\langle \Pi\right\rangle
_{\phi_{t}}^{2}=\left\langle \Pi\right\rangle _{\phi_{t}}\left(
1-\left\langle \Pi\right\rangle _{\phi_{t}}\right)  $ it follows that for all
$t\in\mathbb{R}$%
\begin{equation}
\left\vert P_{\phi}^{\prime}\left(  t\right)  \right\vert \leq2\left(  \Delta
h\right)  _{\phi}\sqrt{P_{\phi}\left(  t\right)  \left(  1-P_{\phi}\left(
t\right)  \right)  }.\label{Tamm}%
\end{equation}

We first simplify this inequality by introducing the dimensionless time
variable $x=t\left(  \Delta h\right)  _{\phi}$ and the function $v:\mathbb{R}%
\rightarrow\left[  0,1\right]  $ with $v\left(  x\right)  =P_{\phi}\left(
t\right)  .$ Inequality (\ref{Tamm}) then becomes equivalent to%
\[
-2\sqrt{v\left(  x\right)  \left(  1-v\left(  x\right)  \right)
}\leq v^{\prime}\left(  x\right)  \leq2\sqrt{v\left(  x\right)
\left(  1-v\left( x\right)  \right)  }\mbox{ for all
}x\in\mathbb{R}.
\]

In order to make use of the differential inequality
\begin{equation}
-2\sqrt{v\left(  x\right)  \left(  1-v\left(  x\right)  \right)  }\leq
v^{\prime}\left(  x\right) \label{Tammlowerb}%
\end{equation}
we first discuss the (autonomous) differential equation
\begin{equation}
y^{\prime}=f\left(  x,y\right)  \mbox{ with
}f:\mathbb{R}\times\left( 0,1\right)  \rightarrow\mathbb{R},\
f\left(  x,y\right)  =-2\sqrt{y\left(
1-y\right)  }.\label{Diffequ}%
\end{equation}

The function $y_{0}:\left(  0,\pi/2\right)  \rightarrow\left(  0,1\right)  $
with $y_{0}\left(  x\right)  =\cos^{2}x$ is a solution of this differential
equation since for all $x\in\left(  0,\pi/2\right)  $%
\[
y_{0}^{\prime}\left(  x\right)  =-2\cos\left(  x\right)  \sin\left(  x\right)
=-2\sqrt{y_{0}\left(  x\right)  }\sqrt{1-y_{0}\left(  x\right)  }=f\left(
x,y_{0}\left(  x\right)  \right)  .
\]
This solution of (\ref{Diffequ}) is of maximal domain since the limits%
\[
\lim_{x\rightarrow0}y_{0}\left(  x\right)  =1\mbox{ and
}\lim_{x\rightarrow \pi/2}y_{0}\left(  x\right)  =0
\]
do not belong the admitted range $0<y<1$ of solutions. Other
solutions of maximal domain are obtained from $y_{0}$ by
translation: $y_{c}\left( x\right)  =y_{0}\left(  x-c\right)  $ for
$c<x<c+\pi/2.$ By a suitable choice of $c$ the initial value problem
$y_{c}\left(  \xi\right)  =\eta$ for any $\left(  \xi,\eta\right)
\in\mathbb{R}\times\left(  0,1\right)  $ is solved. Since $f$ obeys
the local Lipschitz condition of the uniqueness theorem for the
solutions of first order differential equations, the set of all
solutions to $y^{\prime}=f\left(  x,y\right)  $ with maximal domain
is given by $\left\{  y_{c}\left\vert c\in\mathbb{R}\right.
\right\}  .$\\

The continuous extension $g$ of $f$ to the domain $\mathbb{R}\times\left[
0,1\right]  $ leads to the differential equation $z^{\prime}=g\left(
x,z\right)  =-2\sqrt{z\left(  1-z\right)  }$ which violates the local
Lipschitz condition on the boundary points $\left(  x,z\right)  $ with either
$z=0$ or $z=1.$ The set of solutions of the extended equation with maximal
domain is given by $\left\{  z_{c}\left\vert c\in\mathbb{R}\right.  \right\}
$ with
\[
z_{c}:\mathbb{R}\rightarrow\mathbb{R},\, z_{c}\left(  x\right)
=\left\{
\begin{array}[c]{ll}
1 & \mbox{for }x<c\\
\cos^{2}\left(  x-c\right)  & \mbox{for }c\leq x\leq c+\pi/2\\
0 & \mbox{for }x>c+\pi/2
\end{array}
\right.
\]

Thus any function $z_{c}$ with $c\geq0$ is a solution of the initial
value problem $z\left(  0\right)  =1$ with maximal domain. For any
such solution $z_{c}$ with $c\geq0$\ holds
\[
z_{0}\left(  x\right)  \leq z_{c}\left(  x\right)  \leq1
\]
for all $x\geq0.$\\

According to a theorem of differential inequalities, quoted in the appendix A,
we then conclude from (\ref{Tammlowerb}) and from $v\left(  0\right)  =1$ that
for all $x\geq0$%
\[
v\left(  x\right)  \geq z_{0}\left(  x\right)  .
\]
Thus $v\left(  x\right)  \geq\cos^{2}x$ for all $x\in\left[
0,\pi/2\right] .$ This is Fleming's inequality.\\

Suppose now that $\eta:=v(\xi)>\cos^{2}\xi$ for some $\xi\in\left(
0,\pi/2\right)  .$ With $\eta=\cos^{2}\left(  \xi-c\right)  $ for
some $c\in\left(  0,\pi/2\right)  $ it then follows again from the
quoted theorem on differential inequalities that $v\left(  x\right)
\geq\cos^{2}\left( x-c\right)  >\cos^{2}\left(  x\right)  $ for all
$x\in\left[  \xi ,\pi/2\right]  .$ From Fleming's inequality we now
have the two cases only:

\begin{enumerate}
\item For any $\varepsilon>0$ there exists a $\xi\in\left(
0,\varepsilon \right)  $ with $v(\xi)>\cos^{2}\xi.$

\item There exists an $\varepsilon>0$ with $v\left(  x\right)
=\cos^{2}x$ for all $x\in\left(  0,\varepsilon\right)  .$
\end{enumerate}

In the case (i) we have $v\left(  x\right)  \geq\cos^{2}\left(  x-c\right)
>\cos^{2}\left(  x\right)  $ for all $x\in\left[  \xi,\pi/2\right]  .$ Since
there exist such $\xi$ arbitrarily close to $0$ it follows that
$v\left( x\right)  >\cos^{2}\left(  x\right)  $ for all $x\in\left(
0,\pi/2\right]  .$ Since $v$ is an even function the inequality
extends to all $x$ with $\left\vert x\right\vert \in\left(
0,\pi/2\right]  .$\\

In case of (ii) the identity theorem of holomorphic functions
implies $v\left(  x\right)  =\cos^{2}\left(  x\right)  $ for all
$x\in\mathbb{R}$ since $v$ is the restriction of an entire
function to the real line. Thus we have derived the alternatives
(i) and (ii) as being exhaustive.\\

Suppose now that alternative (ii) holds. From the spectral
decomposition (\ref{spectrep}) of $P_{\phi}$ we extract the
constant term and the one with
the highest frequency according to%
\begin{eqnarray*}
P_{\phi}\left( t\right)
&=&\sum_{\alpha=1}^{n}\lambda_{\alpha}^{2}+2\sum_{\hidewidth\substack{\alpha,\beta=1
\atop
\alpha>\beta}\hidewidth}^{n}\lambda_{\alpha}\lambda_{\beta}\cos\left[
\left(
\omega_{\alpha}-\omega_{\beta}\right)  t\right] \\
&=&\sum_{\alpha=1}^{n}\lambda_{\alpha}^{2}+2\lambda_{n}\lambda_{1}\cos\left[
\left(  \omega_{n}-\omega_{1}\right)  t\right] + 2
\sum_{\hidewidth\substack{ \alpha ,\beta=1 \atop \alpha>\beta,
\left( \alpha,\beta\right) \neq\left(n,1\right) }\hidewidth}^n
\lambda_{\alpha}\lambda_{\beta}\cos\left[  \left( \omega_{\alpha
}-\omega_{\beta}\right)  t\right]  .
\end{eqnarray*}
The assumption $P_{\phi}\left(  t\right)  =\cos^{2}\left(  \left(  \Delta
h\right)  _{\phi}t\right)  =\frac{1}{2}\left(  1+\cos\left(  2\left(  \Delta
h\right)  _{\phi}t\right)  \right)  $ now implies, due to $\lambda_{\alpha
}\lambda_{\beta}>0$ for all $\alpha,\beta,$ that the index set of the last sum
is empty. Thus we have $n=2$ and
\[
\lambda_{1}^{2}+\lambda_{2}^{2}=\frac{1}{2},\ 2\lambda_{1}\lambda_{2}=\frac
{1}{2},\ \omega_{2}-\omega_{1}=2\left(  \Delta h\right)  _{\phi}.
\]
The first two equations imply $\lambda_{1}=\lambda_{2}=1/2.$ From this it
follows that the third condition $\omega_{2}-\omega_{1}=2\left(  \Delta
h\right)  _{\phi}$ holds, since%
\begin{eqnarray*}
\left(  \Delta h\right)  _{\phi}^{2}  & =\lambda_{1}\omega_{1}^{2}+\lambda
_{2}\omega_{2}^{2}-\left(  \lambda_{1}\omega_{1}+\lambda_{2}\omega_{2}\right)
^{2}\\
& =\frac{1}{2}\left(  \omega_{1}^{2}+\omega_{2}^{2}\right)  -\frac{1}%
{4}\left(  \omega_{1}+\omega_{2}\right)  ^{2}\\
& =\frac{1}{4}\left(  \omega_{1}-\omega_{2}\right)  ^{2}.
\end{eqnarray*}

Thus we have derived from alternative (ii) that $\phi$ is a linear
combination of just two eigenvectors of $h$ with spectral
components of equal norm. The inverse conclusion that alternative
(ii) follows from $\phi=\phi_{1}+\phi_{2}$ with
$h\phi_{i}=\omega_{i}\phi_{i},\omega_{2}>\omega_{1}$ and
$\left\Vert \phi_{i}\right\Vert ^{2}=\lambda_{i}=1/2$ is obvious
from
\begin{eqnarray*}
P_{\phi}\left(  t\right)   & =\lambda_{1}^{2}+\lambda_{2}^{2}+2\lambda
_{1}\lambda_{2}\cos\left[  \left(  \omega_{2}-\omega_{1}\right)  t\right] \\
& =\frac{1}{2}\left(  1+\cos\left[  \left(  \omega_{2}-\omega_{1}\right)
t\right]  \right)  =\cos^{2}\left(  \left(  \Delta h\right)  _{\phi}t\right)
.
\end{eqnarray*}

\end{proof}

\section{Mixed state decay}

Let $\rho:\mathcal{H}\rightarrow\mathcal{H}$ be a density operator
on the finite dimensional Hilbert space $\mathcal{H},$ i.e. $\rho$
is linear with $\rho\geq0$ and $\Tr\left(  \rho\right)  =1.$ Due to
the spectral theorem there exist mutually orthogonal vectors
$\psi_{1},\ldots\psi_{n}$ with $\left\Vert \psi_{k}\right\Vert =1$
for all $k$ and there exist numbers $\lambda
_{1},\ldots\lambda_{n}\in\mathbb{R}_{>0}$ with
$\sum_{k=1}^{n}\lambda_{k}=1$
such that%
\begin{equation}
\rho=\sum_{k=1}^{n}\lambda_{k}\psi_{k}\left\langle \psi_{k},\cdot\right\rangle
.\label{spectrepdenop}%
\end{equation}

The orthogonal projection $\Pi:\mathcal{H}\rightarrow\mathcal{H}$ onto the
range of $\rho$ is given by
\[
\Pi=\sum_{k=1}^{n}\psi_{k}\left\langle \psi_{k},\cdot\right\rangle .
\]
The projection $\Pi$ is the smallest orthogonal projection with
$\Tr\left( \rho\Pi\right)  =1.$\\

For an arbitrary orthogonal projection
$E:\mathcal{H}\rightarrow\mathcal{H}$ the nonnegative number
$\Tr\left(  \rho E\right)  $ is the probability that the state
$\rho$ passes a filter for the property associated with $E.$ More
generally, the expectation value of a linear symmetric operator $A:\mathcal{H}%
\rightarrow\mathcal{H}$ is given by $\left\langle A\right\rangle
_{\rho }=\Tr\left(  A\rho\right)  $ and its variance reads $\left(
\Delta A\right) _{\rho}^{2}=\left\langle A^{2}\right\rangle
_{\rho}-\left\langle A\right\rangle _{\rho}^{2}.$\\

The dynamics $\phi\mapsto\phi_{t}=\exp\left(  -\rmi h t\right) \phi$
is extended from vectors to density operators through
$\rho\mapsto\rho_{t}=\rme^{-\rmi ht}\rho \rme^{\rmi h t}.$ As a
generalization of the survival probability to mixed states one may
consider the function $P_{\rho}:\mathbb{R}\rightarrow\left[
0,1\right]  $ with
\[
P_{\rho}\left(  t\right)  =\Tr\left(  \Pi \rme^{-\rmi ht}\rho
\rme^{\rmi ht}\right) =\left\langle \rme^{\rmi ht}\Pi \rme^{-\rmi
ht}\right\rangle _{\rho}=\left\langle \Pi\right\rangle _{\rho_{t}}.
\]
The number $P_{\rho}\left(  t\right)  $ thus gives the probability
that the evolved state $\rho_{t}$ has the property $\Pi$ associated
with the initial state $\rho.$ Again $t=0$ is an absolute maximum of
$P_{\rho}$ since $P_{\rho }\left(  0\right)  =1.$ From this it
follows that $P_{\rho}^{\prime}\left( 0\right)  =0$ since $P_{\rho}$
is differentiable.\\

Let $\Phi_{1},\ldots\Phi_{q}$ with $q\geq n$ be an orthonormal basis of
$\mathcal{H}$ such that $h\Phi_{r}=\omega_{r}\Phi_{r}$ for $r=1,\ldots q.$
Then holds%
\begin{eqnarray*}
P_{\rho}\left(  t\right)   & =\Tr\left(  \rme^{\rmi ht}\Pi
\rme^{-\rmi ht}\rho\right) =\sum_{r=1}^{q}\left\langle
\Phi_{r},\rme^{\rmi ht}\Pi \rme^{-\rmi ht}\rho\Phi
_{r}\right\rangle \\
& =\sum_{r,s=1}^{q}\rme^{\rmi \left(  \omega_{r}-\omega_{s}\right)
t}\left\langle \Phi_{r},\Pi\Phi_{s}\right\rangle \left\langle
\Phi_{s},\rho\Phi _{r}\right\rangle .
\end{eqnarray*}
Thus $P_{\rho}$ is a finite linear combination of exponentials and
thus of $C^{\infty}$ type.\\

As in the case of pure states the condition $\left(  \Delta h\right)  _{\rho
}=0$ implies $P_{\rho}\left(  t\right)  =1$ for all $t.$ This can be seen as
follows%
\begin{eqnarray*}
0  & =\left(  \Delta h\right)  _{\rho}^{2}=\left\langle h^{2}\right\rangle
_{\rho}-\left\langle h\right\rangle _{\rho}^{2}=\left\langle \left(
h-\left\langle h\right\rangle _{\rho}\right)  ^{2}\right\rangle _{\rho}\\
& =\sum_{k=1}^{n}\lambda_{k}\left\Vert \left(  h-\left\langle h\right\rangle
_{\rho}\right)  \psi_{k}\right\Vert ^{2}.
\end{eqnarray*}
Thus we have $( h-\left\langle h\right\rangle _{\rho}) \psi _{k}=0$
for all $k.$ Therefore all the vectors $\psi_{k}$ contributing to
the spectral decomposition of $\rho$ are eigenvectors of $h$ (with
the same eigenvalue). From this then follows the stationarity of
$\rho,$ i.e. $\rho _{t}=\rho$ for all $t.$ While in the case of pure
states the condition $\left(  \Delta h\right)_{\phi}>0$ implies that
$P_{\phi}$ is not constant, this is not so with mixed states. A
counterexample is provided by any $\rho$ such that $\Pi$ commutes
with $h$ as it is, e.g., the case for $\rho\left( \mathcal{H}\right)
=\mathcal{H},$ since then
$\Pi=id_{\mathcal{H}}.$\\

In order to better understand $P_{\rho}$ near $0$ we first observe%
\begin{eqnarray*}
P_{\rho}\left(  t\right)   & =\Tr\left(  \Pi \rme^{-\rmi ht}\rho
\rme^{\rmi ht}\right) =\sum_{k=1}^{n}\left\langle \psi_{k},
\rme^{-\rmi ht}\rho \rme^{\rmi ht}\psi
_{k}\right\rangle \\
& =\sum_{k,l=1}^{n}\left\langle \psi_{k},\rme^{-\rmi
ht}\psi_{l}\right\rangle \lambda_{l}\left\langle \psi_{l},\rme^{\rmi
ht}\psi_{k}\right\rangle =\sum _{k,l=1}^{n}\lambda_{l}\left\vert
\left\langle \psi_{k},\rme^{-\rmi ht}\psi
_{l}\right\rangle \right\vert ^{2}\\
& =\sum_{k=1}^{n}\lambda_{k}\left\vert \left\langle
\psi_{k},\rme^{-\rmi ht}\psi
_{k}\right\rangle \right\vert ^{2}+\sum_{\substack{k,l=1 \atop k\neq l}%
}^{n}\lambda_{l}\left\vert \left\langle \psi_{k},\rme^{-\rmi
ht}\psi_{l}\right\rangle \right\vert ^{2}.
\end{eqnarray*}
We thus have%
\begin{equation}
P_{\rho}\left(  t\right)
=\sum_{k=1}^{n}\lambda_{k}P_{\psi_{k}}\left( t\right)
+\sum_{\substack{k,l=1 \atop k\neq l}}^{n}\lambda_{l}\left\vert
\left\langle \psi_{k},\rme^{-\rmi ht}\psi_{l}\right\rangle \right\vert ^{2}%
.\label{mixeddec}%
\end{equation}
The Taylor expansion of $P_{\rho}$ at $0$ now yields%
\begin{eqnarray*}
P_{\rho}\left(  t\right)   & =\sum_{k=1}^{n}\lambda_{k}\left(
1-\left( \Delta h\right)  _{\psi_{k}}^{2}t^{2}\right)
+t^{2}\sum_{\substack{k,l=1 \atop k\neq l}}^{n}\lambda_{l}\left\vert
\left\langle \psi_{k},h\psi
_{l}\right\rangle \right\vert ^{2}+\Or\left(  t^{3}\right) \\
& =1-t^{2}\sum_{k=1}^{n}\lambda_{k}\left(  \Delta h\right)  _{\psi_{k}}%
^{2}+t^{2}\sum_{\substack{k,l=1 \atop k\neq
l}}^{n}\lambda_{l}\left\vert \left\langle
\psi_{k},h\psi_{l}\right\rangle \right\vert ^{2}+\Or\left(
t^{3}\right)  .
\end{eqnarray*}
From this we infer
\begin{equation}
-\frac{P_{\rho}^{\prime\prime}\left(  0\right)
}{2}=\sum_{k=1}^{n}\lambda _{k}\left(  \Delta h\right)
_{\psi_{k}}^{2}-\sum_{\substack{k,l=1 \atop k\neq l
}}^{n}\lambda_{l}\left\vert \left\langle
\psi_{k},h\psi_{l}\right\rangle
\right\vert ^{2}.\label{D2}%
\end{equation}

We shall now prove a generalization of Fleming's bound to the survival
probability of mixed states.

\begin{proposition}
Let $\rho:\mathcal{H}\rightarrow\mathcal{H}$ be a density operator
such that $\left(  \Delta h\right)  _{\rho}>0.$ Then exactly one
of the alternatives (i) and (ii) holds.

\begin{enumerate}
\item $P_{\rho}\left(  t\right)  >\cos^{2}\left(  \left(  \Delta
h\right) _{\rho}t\right)  $ for all $t\in\mathbb{R}$ with $0<\left(
\Delta h\right) _{\rho}\left\vert t\right\vert \leq\pi/2$

\item $P_{\rho}\left(  t\right)  =\cos^{2}\left(  \left(  \Delta
h\right) _{\rho}t\right)  $ for all $t\in\mathbb{R}$
\end{enumerate}

Alternative (ii) holds if and only if there exist two (different)
eigenvalues $\omega_{1},\omega_{2}$ of $h$ such that every vector
$\psi_{k}$ which appears in the spectral decomposition
(\ref{spectrepdenop}) of $\rho$ has a
decomposition $\psi_{k}=\phi_{k,1}+\phi_{k,2}$ with%
\[
h\phi_{k,1}=\omega_{1}\phi_{k,1},\
h\phi_{k,2}=\omega_{2}\phi_{k,2}\mbox{ and }\left\langle
\phi_{k,\varepsilon},\phi_{l,\eta}\right\rangle =
\frac{1}{2}\delta_{k,l}\delta_{\varepsilon,\eta}%
\]
for all $k,l\in\left\{  1,\ldots n\right\}  $ and for all $\varepsilon,\eta
\in\left\{  1,2\right\}  .$
\end{proposition}

\begin{proof}
As in the case of pure states we start from%
\[
\frac{\rmd}{\rmd t}P_{\rho}\left(  t\right) =\frac{\rmd}{\rmd
t}\left\langle \rme^{\rmi ht}\Pi \rme^{-\rmi ht}\right\rangle
_{\rho}=\rmi\left\langle \rme^{\rmi ht}\left[ h,\Pi\right]
\rme^{-\rmi ht}\right\rangle _{\rho}=\rmi\left\langle \left[
h,\Pi\right] \right\rangle _{\rho_{t}}.
\]
The generalized uncertainty relation for the mixed state $\rho_{t}$ applied to
the pair of observables $\left(  h,\Pi\right)  $ reads%
\[
2\left(  \Delta h\right)  _{\rho_{t}}\left(  \Delta\Pi\right)  _{\rho_{t}}%
\geq\left\vert \left\langle \left[  h,\Pi\right]  \right\rangle _{\rho_{t}%
}\right\vert .
\]
From $\Pi^{2}=\Pi$ we obtain $\left(  \Delta\Pi\right)  _{\rho_{t}}%
^{2}=P_{\rho}\left(  t\right)  \left(  1-P_{\rho}\left(  t\right)  \right)  $
and therefrom the estimate%
\[
\left\vert \frac{\rmd}{\rmd t}P_{\rho}\left(  t\right)  \right\vert
\leq2\left( \Delta h\right)  _{\rho}\sqrt{P_{\rho}\left(  t\right)
\left(  1-P_{\rho
}\left(  t\right)  \right)  }%
\]
for all $t\in\mathbb{R}.$\\

The alternatives (i) and (ii) follow from this for $t>0$ in exactly
the same way as in the case of the pure state survival probability
$P_{\phi}.$ Since, however, the mixed state survival probability
$P_{\rho}$ need not be an even function, the case $t<0$ needs a
separate consideration: The case $t<0$ is transformed into the case
$t>0$ by replacing the Hamiltonian $h$ through $-h.$ Since the
variance of $-h$ in the state $\rho$ is the same as that of $h,$ the
alternatives (i) and (ii) hold for for $t<0$ unchanged.\\

Suppose now that alternative (ii) holds. Then $P_{\rho}(t)
=\cos^{2}(( \Delta h)_{\rho}t) =1-t^{2}\left( \Delta h\right)
_{\rho}^{2}+\Or\left(  t^{4}\right) $ for $t\rightarrow0.$ Thus
$-P_{\rho}^{\prime\prime}\left( 0\right) /2=\left(  \Delta h\right)
_{\rho}^{2}$ holds. From equation (\ref{D2}) we then obtain
\begin{equation}
\left(  \Delta h\right)  _{\rho}^{2}=\sum_{k=1}^{n}\lambda_{k}\left(
\Delta h\right)  _{\psi_{k}}^{2}-\sum_{\substack{k,l=1 \atop k\neq
l}}^{n}\lambda _{l}\left\vert \left\langle
\psi_{k},h\psi_{l}\right\rangle \right\vert
^{2}.\label{vari}%
\end{equation}

Now a general result of probability theory says that the variance of
a stochastic variable under a mixture of probability measures is
greater or equal to the mixture of individual variances, or more
specifically applied to the present context it says that
\begin{equation}
\left(  \Delta h\right)  _{\rho}^{2}-\sum_{k=1}^{n}\lambda_{k}\left(  \Delta
h\right)  _{\psi_{k}}^{2}=\sum_{k=1}^{n}\sum_{l=k+1}^{n}\lambda_{k}\lambda
_{l}\left(  \left\langle h\right\rangle _{\psi_{k}}-\left\langle
h\right\rangle _{\psi_{l}}\right)  ^{2}\geq0.\label{varmix}%
\end{equation}
The proof of equation (\ref{varmix}) is given in the appendix. From the
equations (\ref{varmix}), and (\ref{vari}) it thus follows that%
\[
0\geq-\sum_{\substack{k,l=1 \atop k\neq l}}^{n}\lambda_{l}\left\vert
\left\langle \psi_{k},h\psi_{l}\right\rangle \right\vert
^{2}=\sum_{k=1}^{n}\sum _{l=k+1}^{n}\lambda_{k}\lambda_{l}\left(
\left\langle h\right\rangle _{\psi_{k}}-\left\langle h\right\rangle
_{\psi_{l}}\right)  ^{2}\geq0.
\]
Thus both sides of this equation must vanish and $\left\langle \psi_{k}%
,h\psi_{l}\right\rangle =0$ and $\left\langle h\right\rangle _{\psi_{k}%
}=\left\langle h\right\rangle _{\psi_{l}}$ follows for all $\left(
k,l\right)  $ with $k\neq l.$ Furthermore we have%
\[
\left(  \Delta h\right)  _{\rho}^{2}=\sum_{k=1}^{n}\lambda_{k}\left(  \Delta
h\right)  _{\psi_{k}}^{2}.
\]

From (\ref{mixeddec}) it follows for $P_{\rho}\left(  t\right)
=\cos ^{2}\left(  \left(  \Delta h\right)  _{\rho}t\right)
=\frac{1}{2}\left(
1+\cos\left(  2\left(  \Delta h\right)  _{\rho}t\right)  \right)  $ that%
\begin{equation}
\frac{1}{2}\left(  1+\cos\left(  2\left(  \Delta h\right)
_{\rho}t\right) \right) =\sum_{k=1}^{n}\lambda_{k}P_{\psi_{k}}\left(
t\right) +\sum_{\substack{k,l=1 \atop k\neq
l}}^{n}\lambda_{l}\left\vert \left\langle
\psi_{k},\rme^{-\rmi ht}\psi_{l}\right\rangle \right\vert ^{2}.\label{mixeddecsat}%
\end{equation}
This implies that each of the even functions $P_{\psi_{k}}$ is a
real linear combination of the constant function $1$ and $\cos(2(
\Delta h)_{\rho}t)  .$ Thus we have for all $t\in\mathbb{R}$%
\begin{eqnarray*}
P_{\psi_{k}}\left(  t\right)   & =A_{k}+B_{k}\cos\left(  2\left(  \Delta
h\right)  _{\rho}t\right)  =A_{k}+B_{k}-2B_{k}\sin^{2}\left(  \left(  \Delta
h\right)  _{\rho}t\right) \\
& =1-2B_{k}\sin^{2}\left(  \left(  \Delta h\right)  _{\rho}t\right)  .
\end{eqnarray*}
with constants $A_{k},B_{k}\in\mathbb{R}$ such that
$P_{\psi_{k}}\left( 0\right)  =A_{k}+B_{k}=1.$ From $0\leq
P_{\psi_{k}}\left(  t\right)  \leq1$ it follows that
$0\leq2B_{k}\leq1.$\\

Thus $P_{\psi_{k}}$ obeys for $t\rightarrow0$%
\[
P_{\psi_{k}}\left(  t\right)  =1-2B_{k}\left(  \Delta h\right)  _{\rho}%
^{2}t^{2}+\Or\left(  t^{4}\right)  .
\]
Taking into account that $\left\langle \psi_{k},h\psi_{l}\right\rangle =0$ for
$k\neq l$ the right hand side of equation (\ref{mixeddecsat}) obeys%
\[
\sum_{k=1}^{n}\lambda_{k}P_{\psi_{k}}\left(  t\right)
+\sum_{\substack{k,l=1 \atop k\neq l}}^{n}\lambda_{l}\left\vert
\left\langle \psi_{k},\rme^{-\rmi ht}\psi _{l}\right\rangle
\right\vert ^{2}=\sum_{k=1}^{n}\lambda_{k}\left( 1-2B_{k}\left(
\Delta h\right) _{\rho}^{2}t^{2}\right)  +\Or\left( t^{4}\right)  .
\]
Thus we conclude from equation (\ref{mixeddecsat}) that%
\[
1-\left(  \Delta h\right)  _{\rho}^{2}t^{2}=\sum_{k=1}^{n}\lambda_{k}\left(
1-2B_{k}\left(  \Delta h\right)  _{\rho}^{2}t^{2}\right)  .
\]

From this it follows that $\sum_{k=1}^{n}\lambda_{k}2B_{k}=1,$
which in turn implies by means of $0\leq2B_{k}\leq1$ that
$2B_{k}=1$ for all $k.$ Thus we have $\left(  \Delta
h\right)_{\psi_{k}}=\left( \Delta h\right)  _{\rho}$
and%
\[
P_{\psi_{k}}\left(  t\right)  =\cos^{2}\left(  \left(  \Delta h\right)
_{\rho}t\right)
\]
for each $k.$ From (\ref{mixeddecsat}) it now follows that%
\[
\sum_{\substack{k,l=1 \atop k\neq l}}^{n}\lambda_{l}\left\vert
\left\langle \psi_{k},\rme^{-\rmi ht}\psi_{l}\right\rangle
\right\vert ^{2}=0
\]
for all $t.$ For each of the vectors $\psi_{k}$ alternative (ii)
of proposition
\ref{Prop1} is thus realized. From $\left\langle h\right\rangle _{\psi_{k}%
}=\left\langle h\right\rangle _{\psi_{l}}$ and from $\left(  \Delta
h\right) _{\psi_{k}}=\left(  \Delta h\right)  _{\rho}$ it finally
follows that the eigenvalues $\omega_{k,\varepsilon}$ in
$h\phi_{k,\varepsilon}=\omega _{k,\varepsilon}\phi_{k,\varepsilon}$
do not depend on $k.$\\

The inverse statement is obvious by direct computation.
\end{proof}

\section{Appendix: Differential inequalities}

Let $I,J$ be two closed real intervals with $\left(  \xi,\eta\right)  \in
I\times J$ and let $f:I\times J\rightarrow\mathbb{R}$ be continuous. Then the
following results can be found in either Chapt. I, \S 9, sects. VI and VIII
(pp. 73 - 75) of Ref. \cite{Wal1} or in Chapt. II, \S 8, sects. IX and X (pp.
67 - 69) of Ref. \cite{Wal2}.

\begin{proposition}
The initial value problem $y\left(  \xi\right)  =\eta$ of the differential
equation $y^{\prime}=f\left(  x,y\right)  $ has two solutions $y_{\ast}$ and
$y^{\ast}$ which both extend to the boundary of $I\times J$ such that any
other solution $y$ of this initial value problem obeys $y_{\ast}\left(
x\right)  \leq y\left(  x\right)  \leq y^{\ast}\left(  x\right)  $ wherever
both sides of an inequality are defined.\footnote{The solution $y_{\ast}$ is
called minimal and $y^{\ast}$ is called maximal. Yet it is also common to call
any solution of maximal domain a maximal solution. These two notions of
maximal solutions thus should not be confused.}
\end{proposition}

\begin{proposition}
Let $v:I\rightarrow J$ and $w:I\rightarrow J$ be $C^{1}$ functions with
\begin{eqnarray*}
v\left(  \xi\right)   & \leq\eta\mbox{ and }v^{\prime}\left(
x\right)  \leq
f\left(  x,v\left(  x\right)  \right)  \mbox{ for all }x\geq\xi\\
w\left(  \xi\right)   & \geq\eta\mbox{ and }w^{\prime}\left(
x\right)  \geq f\left(  x,w\left(  x\right)  \right)  \mbox{ for all
}x\geq\xi
\end{eqnarray*}
then holds $v\left(  x\right)  \leq y^{\ast}\left(  x\right)  $ and $w\left(
x\right)  \geq y_{\ast}\left(  x\right)  $ for all $x\geq\xi$ wherever both
sides of an inequality are defined.
\end{proposition}

\section{Appendix: Variance and mixing}

\begin{lemma}
Let $\rho:\mathcal{H}\rightarrow\mathcal{H}$ be a density operator on the
finite dimensional Hilbert space $\mathcal{H}$ with its spectral decomposition
as given by equation (\ref{spectrepdenop}). Let $h:\mathcal{H}\rightarrow
\mathcal{H}$ be linear and symmetric. We abbreviate $\left\langle
h\right\rangle _{\psi_{k}}$ by $\left\langle h\right\rangle _{k}.$ Then holds%
\[
\left(  \Delta h\right)  _{\rho}^{2}=\sum_{k=1}^{n}\lambda_{k}\left(  \Delta
h\right)  _{k}^{2}+\frac{1}{2}\sum_{k,l=1}^{n}\lambda_{k}\lambda_{l}\left(
\left\langle h\right\rangle _{k}-\left\langle h\right\rangle _{l}\right)
^{2}.
\]

\end{lemma}

\begin{proof}
First we observe that%
\begin{eqnarray*}
\left(  \Delta h\right)  _{\rho}^{2}  & =\left\langle h^{2}\right\rangle
_{\rho}-\left\langle h\right\rangle _{\rho}^{2}=\sum_{k=1}^{n}\lambda
_{k}\left\langle h^{2}\right\rangle _{k}-\sum_{k,l=1}^{n}\lambda_{k}%
\lambda_{l}\left\langle h\right\rangle _{k}\left\langle h\right\rangle _{l}\\
& =\sum_{k=1}^{n}\lambda_{k}\left(  \Delta h\right)  _{k}^{2}+\sum_{k=1}%
^{n}\lambda_{k}\left\langle h\right\rangle _{k}^{2}-\sum_{k,l=1}^{n}%
\lambda_{k}\lambda_{l}\left\langle h\right\rangle _{k}\left\langle
h\right\rangle _{l}.
\end{eqnarray*}
From the last term we extract the contribution with $k=l$ to obtain
for $M:=\left(  \Delta h\right)
_{\rho}^{2}-\sum_{k=1}^{n}\lambda_{k}\left(
\Delta h\right)  _{k}^{2}$%
\[
M=\sum_{k=1}^{n}\lambda_{k}\left\langle h\right\rangle _{k}^{2}-\sum_{k=1}%
^{n}\lambda_{k}^{2}\left\langle h\right\rangle _{k}^{2}-\sum_{k=1}^{n}%
\sum_{l=1,l\neq k}^{n}\lambda_{k}\lambda_{l}\left\langle h\right\rangle
_{k}\left\langle h\right\rangle _{l}.
\]
In the second sum of this we replace $\lambda_{k}^{2}=\lambda_{k}\left(
1-\sum_{l\neq k}\lambda_{l}\right)  $ which yields%
\begin{eqnarray*}
M  & =\sum_{k=1}^{n}\sum_{l=1,l\neq k}^{n}\lambda_{k}\lambda_{l}\left\langle
h\right\rangle _{k}^{2}-\sum_{k=1}^{n}\sum_{l=1,l\neq k}^{n}\lambda_{k}%
\lambda_{l}\left\langle h\right\rangle _{k}\left\langle h\right\rangle _{l}\\
& =\sum_{k=1}^{n}\sum_{l=1,l\neq k}^{n}\lambda_{k}\lambda_{l}\left(
\left\langle h\right\rangle _{k}^{2}-\left\langle h\right\rangle
_{k}\left\langle h\right\rangle _{l}\right) \\
& =\frac{1}{2}\sum_{k=1}^{n}\sum_{l=1,l\neq k}^{n}\lambda_{k}\lambda
_{l}\left(  \left\langle h\right\rangle _{k}^{2}+\left\langle h\right\rangle
_{l}^{2}-2\left\langle h\right\rangle _{k}\left\langle h\right\rangle
_{l}\right) \\
& =\frac{1}{2}\sum_{k=1}^{n}\sum_{l=1,l\neq k}^{n}\lambda_{k}\lambda
_{l}\left(  \left\langle h\right\rangle _{k}-\left\langle h\right\rangle
_{l}\right)  ^{2}.
\end{eqnarray*}

\end{proof}


\section{References}

\begin{thebibliography}{99}                                                                                               %
\bibitem {Exn}P Exner, Open Quantum Systems an Feynman Integrals, Dordrecht,
D. Reidel, 1985

\bibitem {MaT}L I Mandelstam, I E Tamm, The uncertainty Relation between
Energy and Time in Non\-relativistic Quantum Mechanics, Journ Phys
(USSR) 9 (1945) 249 - 54

\bibitem {Fle}G N Fleming, A Unitary Bound on the Evolution of Nonstationary
States, Nuovo Cim 16A (1973) 232 - 40

\bibitem {GLM}V Giovanetti, S Lloyd, L Maccone, Quantum Limits to Dynamical
Evolution, Phys Rev A 67 (2003) 0521091 - 7

\bibitem {Bro}D C Brody, Elementary Derivation of Passage Times, J. Phys. A 36
(2003) 5587 - 93

\bibitem {And}M Andrews, Bounds to Unitary Evolution, Phys Rev A 75 (2007)
0621121 - 2

\bibitem {AnA}J Anandan, Y Aharonov, Geometry of Quantum Evolution, Phys Rev
Lett 65 (1990) 1697 - 700

\bibitem {HoM}N Horesh, A Mann, Intelligent States for the Anandan-Aharonov
Parameter-Based Uncertainty Relation, Journ Phys A 31 (1998) L609 - 11

\bibitem {Sch}L S Schulmann, in Time in Quantum Mechanics, Eds J G Muga et al,
Berlin, Springer 2002

\bibitem {KoZ}P Kosi\'{n}ski, M Zych, Elementary Proof of the Bound on the
Speed of Quantum Evolution, Phys Rev A 73 (2006) 0243031 - 2

\bibitem {Wal1}W Walter, Gew\"{o}hnliche Differentialgleichungen, Berlin,
Springer, 1976

\bibitem {Wal2}W Walter, Differential and Integral Inequalities, Berlin,
Springer, 1970
\end{thebibliography}
\end{document}